\documentclass{llncs}
\usepackage{times}
\usepackage{latexsym,amsmath,amssymb,url}

\newtheorem{krule}{Reduction Rule}

\begin{document}

\title{Simultaneously Satisfying Linear Equations Over $\mathbb{F}_2$: MaxLin2 and Max-$r$-Lin2 Parameterized Above Average}

\author{
Robert Crowston\inst{1} \and Michael Fellows\inst{2} \and Gregory Gutin\inst{1} \and Mark Jones\inst{1}
\and Frances Rosamond\inst{2} \and St{\'e}phan Thomass{\'e}\inst{3} \and Anders Yeo\inst{1}}

\institute{Royal Holloway, University of London\\
  Egham, Surrey TW20 0EX, UK\\
  \and Charles Darwin University\\
Darwin, Northern Territory 0909 Australia\and LIRMM-Universit{\'e} Montpellier II\\34392 Montpellier Cedex, France}
\date{}
\maketitle

\newenvironment{compress}{\baselineskip=10pt}{\par}
\begin{abstract}
\noindent
In the parameterized problem \textsc{MaxLin2-AA}[$k$], we are given a system with variables $x_1,\ldots ,x_n$
consisting of equations of the form $\prod_{i \in I}x_i = b$, where
$x_i,b \in \{-1, 1\}$ and $I\subseteq [n],$ each equation has a positive integral weight, and we are to decide
whether it is possible to simultaneously satisfy
equations of total weight at least $W/2+k$, where $W$ is the total weight of all equations and $k$ is the parameter
(if $k=0$, the possibility is assured).
We show that \textsc{MaxLin2-AA}[$k$] has a kernel with at most $O(k^2\log k)$ variables and can be solved in time
$2^{O(k\log k)}(nm)^{O(1)}$. This solves
an open problem of Mahajan et al. (2006).

The problem \textsc{Max-$r$-Lin2-AA}[$k,r$] is the same as \textsc{MaxLin2-AA}[$k$] with two differences:
each equation has at most $r$ variables and $r$ is the second parameter. We prove a theorem on
\textsc{Max-$r$-Lin2-AA}[$k,r$]
which implies that \textsc{Max-$r$-Lin2-AA}[$k,r$] has a kernel with at most $(2k-1)r$ variables
improving a number of results including
one by Kim and Williams (2010).  The theorem also implies a lower bound on the maximum of a
function $f:\ \{-1,1\}^n \rightarrow \mathbb{R}$
of degree $r$. We show applicability of the lower bound by
giving a new proof of the Edwards-Erd{\H o}s bound (each connected graph on $n$ vertices and $m$ edges
has a bipartite subgraph with at least $m/2 + (n-1)/4$ edges) and obtaining a generalization.
\end{abstract}

\pagenumbering{arabic}
\pagestyle{plain}

\section{Introduction}\label{section:intro}

\noindent {\bf 1.1 MaxLin2-AA and Max-$r$-Lin2-AA.}
While {\sc MaxSat} and its special case {\sc Max-$r$-Sat} have been widely studied in the literature
on algorithms and complexity for many years, {\sc MaxLin2} and its special case {\sc Max-$r$-Lin2} are less known,
but H\aa stad  \cite{Hastad01} succinctly summarized the importance of these two problems by saying that they are
``as basic as satisfiability." These problems provide important tools for the study of constraint satisfaction problems
such as {\sc MaxSat} and {\sc Max-$r$-Sat} since constraint satisfaction problems can often be reduced to {\sc MaxLin2}
or {\sc Max-$r$-Lin2}, see, e.g., \cite{AloGutKimSzeYeo11,AloGutKri04,CroGutJon10,CroGutJonKimRuz10,Hastad01,KimWil}.
As a result, in the last decade, {\sc MaxLin2} and {\sc Max-$r$-Lin2} have attracted significant attention in  algorithmics.

In the problem \textsc{MaxLin2}, we are given a system $S$ consisting of
$m$ equations in variables $x_1,\ldots ,x_n$, where each equation is $\prod_{i \in I_j}x_i = b_j$ and
$x_i,b_j \in \{-1, 1\}$, $j=1,\ldots ,m$.
Equation $j$ is assigned a positive integral weight $w_j$ and we wish to find an assignment of
values to the variables in order to maximize the total weight of the satisfied equations.

Let $W$ be the sum of the weights of all equations in $S$
and let ${\rm sat}(S)$ be the maximum total weight of equations that can be satisfied simultaneously.
To see that $W/2$ is a tight lower bound on ${\rm sat}(S)$ choose assignments to the
variables independently and uniformly at random. Then $W/2$ is the expected weight of satisfied equations
(as the probability of each equation being satisfied is $1/2$) and thus $W/2$ is a lower bound; to see
the tightness consider a system consisting of pairs of equations
of the form $\prod_{i\in I}x_i=-1,\ \prod_{i\in I}x_i=1$ of the same weight, for some non-empty sets $I \subseteq [n]$.
This leads to the following decision problem:
 \begin{quote}
  {\sc MaxLin2-AA}\\ \nopagebreak
    \emph{Instance:} A system $S$ of equations $\prod_{i \in I_j}x_i = b_j$, where $x_i, b_j \in \{-1,1\}$, $j=1,\ldots ,m$;
equation $j$ is assigned a positive integral weight $w_j$, and a nonnegative integer $k$.\\
      \nopagebreak
    \emph{Question:} ${\rm sat}(S)\ge W/2 + k$?
  \end{quote}
The maximization version of {\sc MaxLin2-AA} (maximize $k$ for which the answer is {\sc Yes}),
has been studied in the literature on approximation algorithms, cf. \cite{Hastad01,HasVen04}.
These two papers also studied the following important special case of {\sc MaxLin2-AA}:
\begin{quote}
  {\sc Max-$r$-Lin2-AA}\\ \nopagebreak
    \emph{Instance:} A system $S$ of equations $\prod_{i \in I_j}x_i = b_j$, where $x_i, b_j \in \{-1,1\}$, $|I_j|\le r$,
$j=1,\ldots ,m$;
equation $j$ is assigned a positive integral weight $w_j$, and a nonnegative integer $k$.\\
      \nopagebreak
    \emph{Question:} ${\rm sat}(S)\ge W/2+k$?
  \end{quote}

H\aa stad \cite{Hastad01} proved that, as a maximization problem, {\sc Max-$r$-Lin2-AA} with any fixed $r\ge 3$
(and hence {\sc MaxLin2-AA}) cannot be approximated within $c$ for any $c>1$ unless P$=$NP (that is, the problem is not in APX unless P$=$NP).
H\aa stad and Venkatesh \cite{HasVen04} obtained some
approximation algorithms for the two problems. In particular, they proved that for {\sc Max-$r$-Lin2-AA} there exist a
constant $c>1$ and a randomized polynomial-time algorithm that, with probability at least 3/4, outputs an assignment
with an approximation ratio of at most $c^r\sqrt{m}.$

The problem {\sc MaxLin2-AA} was first studied in the context of parameterized complexity by Mahajan
et al. \cite{MahajanRamanSikdar09} who naturally took $k$ as the parameter\footnote{We provide basic definitions on parameterized algorithms and complexity in Subsection 1.4 below.}.
We will denote this parameterized problem by {\sc MaxLin2-AA}[$k$]. Despite some progress
\cite{CroGutJon10,CroGutJonKimRuz10,GutKimSzeYeo11}, the complexity of {\sc MaxLin2-AA}[$k$] has remained prominently open
in the research area of ``parameterizing above guaranteed bounds'' that has attracted much recent attention (cf.
\cite{AloGutKimSzeYeo11,BolSco2002,CroGutJon10,CroGutJonKimRuz10,GutKimSzeYeo11,KimWil,MahajanRamanSikdar09})
and that still poses well-known and longstanding open problems (e.g., how difficult is it to determine if a planar graph
has an independent set of size at least $(n/4)+k$?).
One can parameterize {\sc Max-$r$-Lin2-AA} by $k$ for any fixed $r$  (denoted by {\sc Max-$r$-Lin2-AA}[$k$]) or
by both $k$ and $r$ (denoted by {\sc Max-$r$-Lin2-AA}[$k,r$])\footnote{While in the preceding literature only
{\sc MaxLin2-AA}[$k$] was considered, we introduce and study {\sc Max-$r$-Lin2-AA}[$k,r$] in the spirit of
Multivariate Algorithmics as outlined by Fellows \cite{FelIWOCA} and Niedermeier \cite{Niedermeier10}.}.

Define the \emph{excess} for $x^0=(x^0_1,\ldots ,x^0_n)\in \{-1, 1\}^n$ over $S$ to be
$$\varepsilon_S(x^0)=\sum_{j=1}^m c_j\prod_{i\in I_j}x^0_i,\ \mbox{ where } c_j=w_{j}b_j.$$
Note that $\varepsilon_S(x^0)$ is the total weight of  equations satisfied by $x^0$
minus the total weight of equations falsified by $x^0$.
The maximum possible value of $\varepsilon_S(x^0)$ is the {\em maximum excess} of $S$.
H\aa stad and Venkatesh \cite{HasVen04} initiated the study of the excess and further research on the topic
was carried out by Crowston et al. \cite{CroGutJonKimRuz10} who concentrated on {\sc MaxLin2-AA}.
In this paper, we study the maximum excess for {\sc Max-$r$-Lin2-AA}.
Note that the excess is a {\em pseudo-boolean function} \cite{BorHam2002},
i.e., a function that maps $\{-1,1\}^n$ to the set of reals.

\vspace{3mm}

\noindent {\bf 1.2 Main Results and Structure of the Paper.}
The main results of this paper are Theorems \ref{thm:main}
and \ref{lemYes}.
In 2006 Mahajan et al. \cite{MahajanRamanSikdar09} introduced {\sc MaxLin2-AA}[$k$] and asked what is its complexity.
We answer this question in Theorem \ref{thm:main} showing that {\sc MaxLin2-AA}[$k$] admits a kernel with at most $O(k^2\log k)$ variables.
The proof of Theorem \ref{thm:main} is based on the main result in \cite{CroGutJonKimRuz10} and
on a new algorithm for {\sc MaxLin2-AA}[$k$] of complexity $n^{2k}(nm)^{O(1)}$. We also prove that {\sc MaxLin2-AA}[$k$] can be solved in time
$2^{O(k\log k)}(nm)^{O(1)}$ (Corollary \ref{cor1}).
The other main result of this paper, Theorem \ref{lemYes}, gives
a sharp lower bound on the maximum excess for {\sc Max-$r$-Lin2-AA} as follows. Let $S$ be an irreducible system (i.e., a system that cannot
be reduced using Rule \ref{rule1} or \ref{rulerank} defined below) and suppose that each equation contains at most $r$ variables. Let $n\ge (k-1)r+1$ and let $w_{\min}$ be the minimum weight of an equation of $S$. Then, in time $m^{O(1)}$, we can find an assignment $x^0$ to variables of $S$  such that $\varepsilon_S(x^0)\ge k\cdot w_{\min}.$

In Section \ref{sec:H}, we give some reduction rules for {\sc Max-$r$-Lin2-AA}, describe an algorithm
$\cal H$ introduced by Crowston et al. \cite{CroGutJonKimRuz10} and give some properties of the maximum excess,
irreducible systems and Algorithm $\cal H$. In Section \ref{sec:max-lin2}, we prove Theorem \ref{thm:main} and
Corollary \ref{cor1}.
A key tool in our proof of Theorem \ref{lemYes} is a lemma on a so-called sum-free subset in a set of
vectors from $\mathbb{F}_2^n$. The lemma and Theorem \ref{lemYes} are proved in Section \ref{sec:max-r-lin2}.
We prove several corollaries of Theorem \ref{lemYes} in Section \ref{sec:PBEE}. The corollaries are on
parameterized and approximation algorithms as well as on lower bounds for the maxima of pseudo-boolean
functions and their applications in graph theory. Our results on parameterized algorithms improve a number of
previously known results including those of Kim and Williams \cite{KimWil}.
We conclude the paper with Section \ref{sec:op}, where we discuss some open problems.

\vspace{3mm}

\noindent {\bf 1.3 Corollaries of Theorem \ref{lemYes}.} The following results have been obtained for {\sc Max-$r$-Lin2-AA}[$k$] when $r$ is fixed and for {\sc Max-$r$-Lin2-AA}[$k,r$]. Gutin et al. \cite{GutKimSzeYeo11} proved that {\sc Max-$r$-Lin2-AA}[$k$] is fixed-parameter tractable and, moreover, has a kernel with $n\le m = O(k^2)$. This kernel is, in fact, a kernel of {\sc Max-$r$-Lin2-AA}[$k,r$] with $n\le m = O(9^rk^2)$. This kernel for {\sc Max-$r$-Lin2-AA}[$k$] was improved by Crowston et al. \cite{CroGutJonKimRuz10}, with respect to the number of variables,  to $n = O(k\log k)$. For {\sc Max-$r$-Lin2-AA}[$k$], Kim and Williams \cite{KimWil} were the first to obtain a kernel with a linear number of variables, i.e., $n = O(k)$. This kernel is, in fact, a kernel with $n\le r(r+1)k$ for  {\sc Max-$r$-Lin2-AA}[$k,r$]. In this paper, we obtain a kernel with $n\le (2k-1)r$ for {\sc Max-$r$-Lin2-AA}[$k,r$]. As an easy consequence of this result we show that the maximization problem {\sc Max-$r$-Lin2-AA} is in APX if restricted to $m=O(n)$ and the weight of each equation is bounded by a constant. This is in the sharp contrast with the fact mentioned above that for each $r\ge 3$, {\sc Max-$r$-Lin2-AA} is not in APX.

Fourier analysis of pseudo-boolean functions, i.e., functions $f:\ \{-1,1\}^n\rightarrow \mathbb{R}$,
has been used in many areas of computer science (cf. \cite{AloGutKimSzeYeo11,CroGutJonKimRuz10,odonn}).
In Fourier analysis, the Boolean domain is often assumed to be $\{-1,1\}^n$ rather than more usual $\{0,1\}^n$
and we will follow this assumption in our paper. Here we use the following well-known and easy to prove fact \cite{odonn}:
each function $f:\ \{-1,1\}^n\rightarrow \mathbb{R}$ can be uniquely written as
\begin{equation}\label{eq1a}f(x)=\hat{f}(\emptyset) + \sum_{I\in {\cal F}}\hat{f}(I)\prod_{i\in I}x_i.\end{equation}
where ${\cal F}\subseteq \{I:\ \emptyset \neq I\subseteq [n]\}$, $[n]=\{1,2,\ldots ,n\}$ and
$\hat{f}(I)$ are non-zero reals. Formula (\ref{eq1a}) is the Fourier expansion of $f$ and $\hat{f}(I)$ are the Fourier coefficients of $f$.
The right hand size of (\ref{eq1a}) is a polynomial and the degree $\max \{|I|:\ I\in {\cal F}\}$ of this polynomial will be called the {\em degree} of $f$.
In Section \ref{sec:PBEE}, we obtain the following lower bound on the maximum of a pseudo-boolean function $f$ of degree $r$:
\begin{equation}\label{LBpb1i}\max_x f(x)\ge \hat{f}(\emptyset) + \lfloor ({\rm rank} A +r -1)/r\rfloor \cdot \min \{|\hat{f}(I)|: I\in {\cal F}\},\end{equation} where $A$ is a $(0,1)$-matrix with entries $a_{ij}$ such that $a_{ij}=1$ if and only if term $j$ in (\ref{eq1a}) contains $x_i$
(as ${\rm rank} A$ does not depend on the order of the columns in $A$, we may order the terms in (\ref{eq1a}) arbitrarily).

To demonstrate the combinatorial usefulness of (\ref{LBpb1i}), we apply it to obtain a short proof of the well-known lower bound of Edwards-Erd{\H o}s on the maximum size of a bipartite subgraph in a graph (the {\sc Max Cut} problem). Erd{\H o}s \cite{Erdos1965} conjectured and Edwards \cite{Edwards1975} proved that every connected graph with $n$ vertices and $m$ edges has a bipartite subgraph with at least $m/2 + (n-1)/4$ edges. For short graph-theoretical proofs, see, e.g., Bollob{\'a}s and Scott \cite{BolSco2002} and Erd{\H o}s et al. \cite{ErdGyaKoh1997}. We consider  the {\sc Balanced Subgraph} problem
\cite{BocHufTruWah2009} that generalizes {\sc Max Cut} and show that our proof of the Edwards-Erd{\H o}s bound can be easily extended to {\sc Balanced Subgraph}, but the graph-theoretical proofs of the Edwards-Erd{\H o}s bound do not seem to be easily extendable to {\sc Balanced Subgraph}.


\noindent {\bf 1.4 Parameterized Complexity and (Bi)kernelization.}
A \emph{parameterized problem} is a subset $L\subseteq \Sigma^* \times
\mathbb{N}$ over a finite alphabet $\Sigma$. $L$ is
\emph{fixed-parameter tractable} (FPT, for short) if the membership of an instance
$(x,k)$ in $\Sigma^* \times \mathbb{N}$ can be decided in time
$f(k)|x|^{O(1)},$ where $f$ is a function of the
parameter $k$ only.
When the decision time is replaced by the much more powerful $|x|^{O(f(k))},$
we obtain the class XP, where each problem is polynomial-time solvable
for any fixed value of $k.$ There is an infinite number of parameterized complexity
classes between FPT and XP (for each integer $t\ge 1$, there is a class W[$t$]) and they form the following tower:
$FPT \subseteq W[1] \subseteq W[2] \subseteq \cdots \subseteq XP.$
For the definition of the classes W[$t$],
see, e.g., \cite{FlumGrohe06}.

Given a pair $L,L'$ of parameterized problems,
a \emph{bikernelization from $L$ to $L'$} is a polynomial-time
algorithm that maps an instance $(x,k)$ to an instance $(x',k')$ (the
\emph{bikernel}) such that (i)~$(x,k)\in L$ if and only if
$(x',k')\in L'$, (ii)~ $k'\leq f(k)$, and (iii)~$|x'|\leq g(k)$ for some
functions $f$ and $g$. The function $g(k)$ is called the {\em size} of the bikernel.
The notion of a bikernelization was introduced in \cite{AloGutKimSzeYeo11}, where it was observed that
a parameterized problem $L$ is fixed-parameter
tractable if and only if it is decidable and admits a
bikernelization from itself to a parameterized problem $L'$.
A {\em kernelization} of a parameterized problem
$L$ is simply a bikernelization from $L$ to itself; the bikernel is the {\em kernel}, and $g(k)$ is the {\em size} of
the kernel. Due to the importance of polynomial-time kernelization algorithms in applied multivariate algorithmics,
low degree polynomial size kernels and bikernels are of considerable interest, and the subject has developed
substantial theoretical depth, cf.
\cite{AloGutKimSzeYeo11,BodlaenderEtAl2009a,BodJanKra2011,BodlaenderEtAl2009,DLS09,FlumGrohe06,FomLokMisPhiSau2011,GutIerMniYeo,GutKimSzeYeo11}.

The case of several parameters $k_1,\ldots ,k_t$ can be reduced to the one parameter case by setting $k=k_1+\cdots +k_t,$ see, e.g., \cite{DLS09}.

\section{Maximum Excess, Irreducible Systems and Algorithm $\cal H$}\label{sec:H}

Recall that an instance of {\sc MaxLin2-AA} consists of a system $S$ of equations
$\prod_{i \in I_j}x_i = b_j,$ $j\in [m]$, where $\emptyset \neq I_j\subseteq [n]$, $b_j \in \{-1, 1\}$, $x_i\in \{-1, 1\}.$ An equation $\prod_{i \in I_j}x_i = b_j$ has an integral positive weight $w_{j}$. Recall that the excess for $x^0=(x^0_1,\ldots ,x^0_n)\in \{-1, 1\}^n$ over $S$ is $\varepsilon_S(x^0)=\sum_{j=1}^m c_j\prod_{i\in I_j}x^0_i$, where $c_j=w_{j}b_j$. The excess $\varepsilon_S(x^0)$ is the total weight of  equations satisfied by $x^0$
minus the total weight of equations falsified by $x^0$. The maximum possible value of $\varepsilon_S(x^0)$ is the maximum excess of $S$.

\begin{remark} \label{MaxLinExcess} Observe that the answer to {\sc MaxLin2-AA} is {\sc Yes} if and only if
the maximum excess is at least $2k$. \end{remark}

\begin{remark} \label{PseudoBExcess} The excess $\varepsilon_S(x)$ is a pseudo-boolean function and its Fourier expression
is  $\varepsilon_S(x)=\sum_{j=1}^m c_j\prod_{i\in I_j}x_i$. Moreover, observe that every pseudo-boolean function
$f(x)=\sum_{I\in {\cal F}}\hat{f}(I)\prod_{i\in I}x_i$ (where $\hat{f}(\emptyset)=0$) is the excess over the
system $\prod_{i\in I}x_i=b_I$, $I\in {\cal F}$, where $b_I=1$ if $\hat{f}(I)>0$ and $b_I=-1$ if $\hat{f}(I)<0$,
with weights $|\hat{f}(I)|.$ Thus, studying the maximum excess over a {\sc MaxLin2-AA}-system (with real weights)
is equivalent to studying the maximum of a pseudo-boolean function. \end{remark}

Consider two reduction rules for {\sc MaxLin2} studied in \cite{GutKimSzeYeo11}.

  \begin{krule}\label{rule1}
  If we have, for a subset $I$ of $[n]$, an equation $\prod_{i \in I} x_i =b_I'$
  with weight $w_I'$, and an equation $\prod_{i \in I} x_i =b_I''$ with weight $w_I''$,
  then we replace this pair by one of these equations with weight $w_I'+w_I''$ if $b_I'=b_I''$ and, otherwise, by
  the equation whose weight is bigger, modifying its
  new weight to be the difference of the two old ones. If the resulting weight
  is~0, we delete the equation from the system.
  \end{krule}
    \begin{krule}\label{rulerank}
Let $A$ be the matrix over $\mathbb{F}_2$ corresponding to the set of equations in $S$, such that $a_{ji} = 1$ if $i \in I_j$ and $0$, otherwise.
  Let $t={\rm rank} A$ and suppose columns $a^{i_1},\ldots ,a^{i_t}$ of $A$ are linearly independent.
  Then delete all variables not in $\{x_{i_1},\ldots ,x_{i_t}\}$ from the equations of $S$.
  \end{krule}
  \begin{lemma}\label{lem:SS'}\cite{GutKimSzeYeo11}
  Let $S'$ be obtained from $S$ by Rule~\ref{rule1} or \ref{rulerank}.
  Then the maximum excess of $S'$  is equal to the maximum excess of $S$.
  Moreover, $S'$ can be obtained from $S$ in time polynomial in $n$ and $m$.
  \end{lemma}
If we cannot change a weighted system $S$ using  Rules~\ref{rule1} and \ref{rulerank}, we call it {\em irreducible}.

\begin{lemma}\label{lem:irred}
Let $S'$ be a system obtained from $S$ by first applying Rule~\ref{rule1} as long as possible and then Rule \ref{rulerank} as long as possible. Then $S'$ is irreducible.
\end{lemma}
\begin{proof}
Let $S^*$ denote the system obtained from $S$ by applying Rule~\ref{rule1} as long as possible. Without loss of generality, assume that $x_1\not\in \{x_{i_1},\ldots ,x_{i_t}\}$ (see the description of Rule \ref{rulerank}) and thus Rule \ref{rulerank} removes $x_1$ from $S^*$. To prove the lemma
it suffices to show that after $x_1$ removal no pair of equations has the same left hand side. Suppose that there is a pair of equations in $S^*$ which has the same left hand side after $x_1$ removal; let $\prod_{i \in I'} x_i =b'$
and $\prod_{i \in I''} x_i =b''$ be such equations and let $I'=I''\cup \{1\}$. Then the entries of the first column of $A$, $a^1$, corresponding to the pair of equations are 1 and 0, but in all the other columns of $A$ the entries corresponding to the the pair of equations are either 1,1 or 0,0. Thus, $a^1$ is independent from all the other columns of $A$, a contradiction.\qed
\end{proof}

%


Let $S$ be an irreducible system of {\sc MaxLin2-AA}.
Consider the following algorithm introduced in \cite{CroGutJonKimRuz10}.
We assume that, in the beginning, no equation or variable in $S$ is marked.

\begin{center}
\fbox{~\begin{minipage}{11cm}
\textsc{Algorithm $\cal H$}

\smallskip
While the system $S$ is nonempty and the total weight of marked equations is less than $2k$ do the following:

\begin{enumerate}

					                   \item Choose an arbitrary equation  $\prod_{i \in I} x_i =b$ and mark an arbitrary variable $x_l$ such that $l \in I$.
                                       \item Mark this equation and delete it from the system.
                                       \item Replace every equation $\prod_{i \in I'} x_i =b'$ in the system containing $x_l$ by
                                      $\prod_{i \in I\Delta I'} x_i = bb'$, where $I\Delta I'$ is the symmetric difference of $I$ and $I'$ (the weight of the equation is unchanged).
                                       \item Apply Reduction Rule \ref{rule1} to the system.
                                     \end{enumerate}
\smallskip
\end{minipage}~}
\end{center}

Note that algorithm ${\cal H}$ replaces $S$ with an \emph{equivalent}
system under the assumption that the marked equations are satisfied; that
is, for every assignment of values to the variables $x_1, \ldots, x_n$ that
satisfies the marked equations, both systems have the same excess.
As a result, we have the following lemma.

\begin{lemma}\label{lemExcess}\cite{CroGutJonKimRuz10}
Let $S$ be an irreducible system and assume that Algorithm $\cal H$ marks equations of total weight $w$.
Then the maximum excess of $S$ is at least $w$.
In particular, if $w\ge 2k$ then $S$ is a {\sc Yes}-instance of {\sc MaxLin2-AA}[$k$].
\end{lemma}

\section{MaxLin2-AA}\label{sec:max-lin2}

The following two theorems form a basis for proving Theorem \ref{thm:main}, the main result of this section.

\begin{theorem}\label{thm:ing1}
There exists an $n^{2k}(nm)^{O(1)}$-time algorithm for {\sc MaxLin2-AA}[$k$] that returns an
assignment of excess of at least $2k$ if one exists, and returns {\sc no} otherwise.
\end{theorem}

\begin{proof}
Suppose we have an instance $\cal L$ of {\sc MaxLin2-AA}[$k$] that is reduced by Rules \ref{rule1} and \ref{rulerank},
and that the maximum excess of $\cal L$ is at least $2k$. Let $A$ be the matrix introduced
in Rule \ref{rulerank}.
Pick $n$ equations $e_1,\ldots , e_n$ such that their rows in $A$ are
linearly independent. Any assignment must either satisfy one of these
equations, or falsify them all. We can check, in time $(nm)^{O(1)}$, what
happens if they are all falsified, as fixing the values of these $n$
equations fixes the values of all the others. If falsifying all the
equations does not lead to an excess of at least $2k$, then any assignment
of values to $x_1,\ldots ,x_n$ that leads to excess at least $2k$ must satisfy at least
one of $e_1, \ldots ,e_n$. Thus, by Lemma \ref{lemExcess}, algorithm $\cal
H$ can mark one of these equations and achieve an excess of
at least $2k$.

This gives us the following depth-bounded search tree. At each node $N$ of the
tree, reduce the system by Rules  \ref{rule1} and \ref{rulerank}, and let $n'$ be the number of variables
in the reduced system. Then
find $n'$ equations $e_1,\ldots ,e_{n'}$ corresponding to linearly independent
vectors. Find an assignment of values to $x_1,\ldots ,x_{n'}$ that falsifies all of $e_1,\ldots , e_{n'}$.
Check whether this assignment achieves excess of at least $2k-w^*$, where $w^*$ is total weight of equations marked
by $\cal H$ in all predecessors of $N$. If it does,
then return the assignment and stop the algorithm.
Otherwise, split into $n'$ branches. In the $i$'th branch, run an iteration
of $\cal H$ marking equation $e_i$. Then repeat this algorithm for each
new node. Whenever the total weight of marked equations is at least $2k,$ return the
suitable assignment. Clearly, the algorithm will terminate without an assignment if the maximum excess of
$\cal L$ is less than $2k.$

All the operations at each node take time $(nm)^{O(1)}$, and there are
less than $n^{2k+1}$ nodes in the search tree.
Therefore this algorithm takes time $n^{2k}(nm)^{O(1)}$.\qed
\end{proof}

\begin{theorem}\label{thm:ing2}\cite{CroGutJonKimRuz10}
Let $S$ be an irreducible system of {\sc MaxLin2-AA}[$k$] and let $k\ge 2.$ If $k\le m\le 2^{n/(k-1)}-2$, then the maximum excess of $S$ is at least $k$. Moreover, we can find an assignment with excess of at least $k$ in time $m^{O(1)}$.
\end{theorem}

\begin{theorem}\label{thm:main}
The problem {\sc MaxLin2-AA}[$k$] has a kernel with at most $O(k^2\log k)$ variables.
\end{theorem}
\begin{proof}
Let $\cal L$ be an instance of {\sc MaxLin2-AA}[$k$] and let $S$ be
the system of $\cal L$ with $m$ equations and $n$ variables.
We may assume that $S$ is irreducible. Let the parameter $k$ be an
arbitrary positive integer.

If $m<2k$ then $n<2k=O(k^2\log k)$. If $2k\le m\le 2^{n/(2k-1)}-2$
then, by Theorem \ref{thm:ing2}, the answer to $\cal L$ is {\sc yes}
and the corresponding assignment can be found in polynomial time. If $m\ge
n^{2k}$ then, by Theorem \ref{thm:ing1}, we can solve $\cal L$ in
polynomial time.

Finally we consider the case $2^{n/(2k-1)}-1 \le m \le n^{2k}-1$.
Hence, $n^{2k}\ge 2^{n/(2k-1)}.$
Therefore, $4k^2\ge 2+n/\log n \ge \sqrt{n}$ and $n\le (2k)^4$. Hence,
$n\le 4k^2\log n\le 4k^2\log(16k^4)=O(k^2\log k).$

Since $S$ is irreducible, $m<2^n$ and thus we have obtained the desired kernel.
\qed
\end{proof}

\begin{corollary}\label{cor1}
The problem {\sc MaxLin2-AA}[$k$] can be solved in time $2^{O(k\log k)}(nm)^{O(1)}$.
\end{corollary}
\begin{proof}
Let $\cal L$ be an instance of {\sc MaxLin2-AA}[$k$]. By Theorem \ref{thm:main}, in time $(nm)^{O(1)}$ either we solve
$\cal L$ or we obtain a kernel
with at most $O(k^2\log k)$ variables. In the second case, we can solve the reduced system (kernel) by the algorithm of
Theorem \ref{thm:ing1} in time
$[O(k^2\log k)]^{2k}[O(k^2\log k)m]^{O(1)}=2^{O(k\log k)}m^{O(1)}$. Thus, the total time is $2^{O(k\log k)}(nm)^{O(1)}.$
\qed \end{proof}

\section{Max-$r$-Lin2-AA}\label{sec:max-r-lin2}

In order to prove Theorem \ref{lemYes}, we will need the following lemma on vectors in $\mathbb{F}^n_2$.
Let $M$ be a set of $m$ vectors in $\mathbb{F}_2^n$ and let $A$ be a $m\times n$-matrix in which the
vectors of $M$ are rows. Using Gaussian elimination on $A$ one can find a maximum size
linearly independent subset of $M$ in polynomial time \cite{KorVyg2006}.
Let $K$ and $M$ be sets of vectors in $\mathbb{F}^n_2$ such that $K \subseteq M$. We say $K$ is \emph{$M$-sum-free}
if no sum of two or more distinct vectors in $K$ is equal to a vector in $M$. Observe that $K$ is $M$-sum-free if
and only if $K$ is linearly independent and no sum of vectors in $K$ is equal to a vector in $M \backslash K$.

\begin{lemma}\label{thmKcon}
Let $M$ be a set of vectors in $\mathbb{F}^n_2$ such that $M$ contains a basis of $\mathbb{F}^n_2.$ Suppose that each vector of $M$ contains at most $r$ non-zero coordinates. If $k\ge 1$ is an integer and $n \ge r(k-1)+1$, then
in time $|M|^{O(1)}$, we can find a subset $K$ of $M$ of $k$ vectors
such that $K$ is $M$-sum-free.
\end{lemma}
\begin{proof}
Let $\mathbf{1} = (1, \ldots, 1)$ be the vector in $\mathbb{F}^n_2$ in which every coordinate is $1$. Note that $\mathbf{1}\not\in M.$
By our assumption $M$ contains a basis of $\mathbb{F}_2^n$ and we may find such a basis in polynomial time (using Gaussian elimination, see above).
We may write $\mathbf{1}$ as a sum of some vectors of this basis $B$. This implies that $\mathbf{1}$ can be expressed as follows: $\mathbf{1} = v_1 + v_2 + \cdots + v_s$, where $\{v_1, \ldots, v_s\}\subseteq B$ and $v_1, \ldots, v_s$ are linearly independent, and we can find such an expression in polynomial time.

For each $v \in M \backslash \{v_1, \ldots, v_s\}$, consider the set $S_v=\{v, v_1, \ldots, v_s\}$. In polynomial time, we may check whether $S_v$ is linearly independent. Consider two cases:

\begin{description}
  \item[Case 1:] $S_v$ is linearly independent for each $v \in M \backslash \{v_1, \ldots, v_s\}$. Then $\{v_1, \ldots, v_s\}$ is $M$-sum-free (here we also use the fact that $\{v_1, \ldots, v_s\}$ is linearly independent). Since each $v_i$ has at most $r$ positive coordinates, we have $sr\ge n> r(k-1)$. Hence, $s>k-1$ implying that $s\ge k$. Thus, $\{v_1, \ldots, v_k\}$ is the required set $K$.
  \item[Case 2:] $S_v$ is linearly dependent for some $v\in M \backslash \{v_1, \ldots, v_s\}$. Then we can find (in polynomial time) $I \subseteq [s]$ such that $v=\sum_{i \in I} v_i$. Thus, we have a shorter expression for $\mathbf{1}$: $\mathbf{1} = v'_1 + v'_2 + \cdots + v'_{s'}$, where $\{v'_1, \ldots , v'_{s'}\}=\{v\} \cup \{ v_i : i \notin I\}$. Note that $\{v'_1, \ldots , v'_{s'}\}$ is linearly independent.
\end{description}

Since $s\le n$ and Case 2 produces a shorter expression for $\mathbf{1}$, after at most $n$ iterations of Case 2 we will arrive at Case 1.
 \qed \end{proof}


Now we can prove the main result of this section.

\begin{theorem}\label{lemYes}
Let $S$ be an irreducible system and suppose that each equation contains at most $r$ variables. Let $n\ge (k-1)r+1$ and let $w_{\min}$ be the minimum weight of an equation of $S$. Then, in time $m^{O(1)}$, we can find an assignment $x^0$ to variables of $S$  such that $\varepsilon_S(x^0)\ge k\cdot w_{\min}.$
\end{theorem}
\begin{proof}
Consider a set $M$ of vectors in $\mathbb{F}_2^n$ corresponding to equations in $S$ as follows: for each equation $\prod_{i \in I} x_i =b$ in $S$, define a vector $v=(v_1,\ldots ,v_{n})\in M$, where $v_i=1$ if $i\in I$ and $v_i=0$, otherwise.

As $S$ is reduced by Rule \ref{rulerank} we have that $M$ contains a basis for $\mathbb{F}^n_2$, and each vector contains at most $r$ non-zero coordinates and $n\ge (k-1)r+1$. Therefore, using Lemma \ref{thmKcon} we can find an $M$-sum-free set $K$ of $k$ vectors. Let $\{e_{j_1}, \ldots, e_{j_k}\}$ be the corresponding set of equations. Run Algorithm $\cal H$, choosing at Step 1
an equation of $S$ from $\{e_{j_1}, \ldots, e_{j_k}\}$ each time,
and let $S'$ be the resulting system.
Algorithm $\cal H$ will run for $k$ iterations of the while loop as
no equation from $\{e_{j_1}, \ldots, e_{j_k}\}$ will be deleted before it has been marked.

Indeed, suppose that this is not true. Then for some $e_{j_l}$ and some other equation $e$ in $S$, after
applying Algorithm $\cal H$ for at most $l-1$ iterations $e_{j_l}$ and $e$ contain the same variables.
Thus, there are vectors $v_j\in K$ and $v\in M$ and a pair of nonintersecting subsets $K'$ and $K''$ of $K\setminus \{v,v_j\}$ such that $v_j+\sum_{u\in K'}u=v+\sum_{u\in K''}u$. Thus,
$v=v_j+\sum_{u\in K'\cup K''}u$, a contradiction with the definition of $K.$

Thus, by Lemma \ref{lemExcess}, we are done.
\qed \end{proof}

\begin{remark}\label{rem:sharp2}
To see that the inequality $n \ge r(k-1)+1$ in the theorem is best possible assume that $n=r(k-1)$ and
consider a partition of $[n]$ into $k-1$ subsets $N_1,\ldots ,N_{k-1}$, each of size $r.$
Let $S$ be the system consisting of subsystems $S_i$, $i\in [k-1],$ such that a subsystem $S_i$ is comprised of equations $\prod_{i \in I} x_i=-1$ of weight 1 for every $I$ such that $\emptyset \neq I\subseteq N_i$. Now assume without loss of generality that $N_i=[r]$. Observe that the assignment $(x_1,\ldots ,x_r)=(1,\ldots ,1)$ falsifies all equations of $S_i$ but by setting $x_j=-1$ for any $j\in [r]$ we satisfy the equation $x_j=-1$ and turn the remaining equations into pairs of the form $\prod_{i \in I} x_i=-1$ and $\prod_{i \in I} x_i=1$. Thus, the maximum excess of $S_i$ is 1 and the maximum excess of $S$ is $k-1$.
\end{remark}

\begin{remark}\label{rem:reals}
It is easy to check that Theorem \ref{lemYes} holds when the weights of equations in $S$ are real
numbers, not necessarily integers.
\end{remark}

\section{Applications of Theorem \ref{lemYes}}\label{sec:PBEE}

\begin{theorem}\label{thLinAAkr}
The problem {\sc Max-$r$-Lin2-AA}[$k,r$] has a kernel with at most $(2k-1)r$ variables.
\end{theorem}

\begin{proof}

Let $T$ be the system of an instance of {\sc Max-$r$-Lin2-AA}[$k,r$]. After applying Rules \ref{rule1} and  \ref{rulerank}  to $T$ as long as possible, we obtain a new system $S$ which is irreducible. Let $n$ be the number of variables in $S$ and observe that the number of variables in an equation in $S$ is bounded by $r$ (as in $T$). If $n \ge (2k-1)r+1$, then, by Theorem \ref{lemYes} and Remark \ref{MaxLinExcess}, $S$ is a {\sc Yes}-instance of {\sc MaxLin2-AA}[$k,r$] and, hence, by Lemma \ref{lem:SS'}, $S$ and $T$ are both {\sc Yes}-instances of {\sc MaxLin2-AA}[$k,r$].
Otherwise $n \le (2k-1)r$ and we have the required kernel.
\qed \end{proof}

\begin{corollary}
The maximization problem {\sc Max-$r$-Lin2-AA} is in APX if restricted to $m=O(n)$ and the weight of each equation is bounded by a constant.
\end{corollary}
\begin{proof}
It follows from Theorem \ref{thLinAAkr} that the answer to {\sc Max-$r$-Lin2-AA}, as a decision problem, is {\sc Yes} as long as $2k\le \lfloor (n+r-1)/r\rfloor.$
This implies approximation ratio at most $W/(2\lfloor (n+r-1)/r\rfloor)$ which is bounded by a constant provided $m=O(n)$ and the weight of each equation is bounded by a constant (then $W=O(n)$).
\qed \end{proof}

The (parameterized) Boolean Max-$r$-Constraint Satisfaction Problem ({\sc Max-$r$-CSP}) generalizes {\sc MaxLin2-AA}[$k,r$]
as follows: We are given a set $\Phi$ of Boolean functions, each involving at
most $r$ variables, and a collection ${\cal F}$ of $m$ Boolean functions, each $f \in \cal F$ being a
member of $\Phi$, and each acting on some subset of
the $n$ Boolean variables $x_1,x_2, \ldots ,x_n$ (each $x_i\in \{-1,1\}$). We are to decide whether there is a truth
assignment to the $n$ variables such that the total number of satisfied functions is at least $E+k$, where  $E$ is the
average value of the number of satisfied functions. The parameters are $k$ and $r$.

Using a bikernelization algorithm described in
\cite{AloGutKimSzeYeo11,CroGutJonKimRuz10} and our new kernel result, it
easy to see that {\sc Max-$r$-CSP} with parameters $k$ and $r$ admits a
bikernel with at most $(k2^{r+1}-1)r$ variables. This result improves the
corresponding result of Kim and Williams \cite{KimWil} ($n\le
kr(r+1)2^r$).

The following result is essentially a corollary of Theorem \ref{lemYes} and Remark \ref{rem:reals}.

\begin{theorem}\label{thm:PBB}
Let \begin{equation}\label{fourier} f(x)=\hat{f}(\emptyset) + \sum_{I\in {\cal F}}\hat{f}(I)\prod_{i\in I}x_i\end{equation} be a pseudo-boolean function of degree $r$. Then
\begin{equation}\label{LBpb1t}\max_x f(x)\ge \hat{f}(\emptyset) + \lfloor ({\rm rank} A +r -1)/r\rfloor \cdot \min \{|\hat{f}(I)|: I\in {\cal F}\},\end{equation} where $A$ is a $(0,1)$-matrix with entries $a_{ij}$ such that $a_{ij}=1$ if and only if term $j$ in (\ref{fourier}) contains $x_i.$ One can find an assignment of values to $x$ satisfying (\ref{LBpb1t}) in time $(n|{\cal F}|)^{O(1)}.$
\end{theorem}
\begin{proof}
By Remark \ref{PseudoBExcess} the function $f(x) - \hat{f}(\emptyset) =\sum_{I\in {\cal F}}\hat{f}(I)\prod_{i\in I}x_i$ is the excess over the system $\prod_{i\in I}x_i=b_I$, $I\in {\cal F}$, where $b_I=+1$ if $\hat{f}(I)>0$ and $b_I=-1$ if $\hat{f}(I)<0$, with weights $|\hat{f}(I)|.$ Clearly, Rule \ref{rule1} will not change the system.
Using Rule \ref{rulerank} we can replace the system by an equivalent one (by Lemma \ref{lem:SS'}) with ${\rm rank} A$ variables. By Lemma \ref{lem:irred}, the new system is irreducible and we can now apply Theorem \ref{lemYes}. By this theorem, Remark \ref{PseudoBExcess} and Remark \ref{rem:reals}, $\max_x f(x)\ge  \hat{f}(\emptyset)+k^*\min \{|\hat{f}(I)|: I\in {\cal F}\}$, where $k^*$ is the maximum value of $k$ satisfying ${\rm rank} A\ge (k-1)r+1$. It remains to observe that $k^*=\lfloor ({\rm rank} A +r -1)/r\rfloor$.
\qed \end{proof}

To give a new proof of the Edwards-Erd{\H o}s bound, we need the following well-known and easy-to-prove fact \cite{BonMur2008}. For a graph $G=(V,E)$, an incidence matrix is a $(0,1)$-matrix with entries $m_{e,v}$, $e\in E$, $v\in V$ such that $m_{e,v}=1$ if and only if $v$ is incident to $e$.

\begin{lemma}\label{lem:M}
The rank of an incident matrix $M$ of a connected graph equals $|V|-1.$
\end{lemma}


\begin{theorem}
Let $G=(V,E)$ be a connected graph with $n$ vertices and $m$ edges. Then $G$ contains a bipartite subgraph with at least $\frac{m}{2} + \frac{n-1}{4}$ edges. Such a subgraph can be found in polynomial time.
\end{theorem}
\begin{proof}
Let $V=\{v_1,v_2,\ldots , v_n\}$ and let $c:\ V\rightarrow \{-1,1\}$ be a 2-coloring of $G$. Observe that the maximum number of edges in a bipartite subgraph of $G$ equals the maximum number of properly colored edges (i.e.,
edges whose end-vertices received different colors) over all 2-colorings of $G$. For an edge $e=v_iv_j\in E$ consider the following function $f_e(x)=\frac{1}{2}(1-x_ix_j)$, where $x_i=c(v_i)$ and $x_j=c(v_j)$ and observe that $f_e(x)=1$ if $e$ is properly colored by $c$ and $f_e(x)=0,$ otherwise. Thus, $f(x)=\sum_{e\in E}f_e(x)$ is the number of properly colored edges for $c$.
We have $f(x)=\frac{m}{2}-\frac{1}{2}\sum_{e\in E}x_ix_j$. By Theorem \ref{thm:PBB}, $\max_x f(x)\ge m/2 + \lfloor ({\rm rank} A +2 -1)/2\rfloor/2.$ Observe that matrix $A$ in this bound is  an incidence matrix of $G$ and, thus, by Lemma \ref{lem:M} ${\rm rank} A = n-1$. Hence, $\max_x f(x)\ge \frac{m}{2} + \frac{1}{2}\lfloor \frac{n}{2} \rfloor \ge \frac{m}{2} + \frac{n-1}{4}$ as required.
\qed \end{proof}

This theorem can be extended to the {\sc Balance Subgraph} problem \cite{BocHufTruWah2009}, where we are given a graph $G=(V,E)$ in which each edge is labeled either by $=$ or by $\neq$ and we are asked to find a 2-coloring of $V$ such that the maximum number of edges is satisfied; an edge labeled by $=$ ($\neq$, resp.) is {\em satisfied} if and only if the colors of its end-vertices
are the same (different, resp.).

\begin{theorem}
Let $G=(V,E)$ be a connected graph with $n$ vertices and $m$ edges labeled by either $=$ or $\neq$. There is a 2-coloring of $V$ that satisfies at least $\frac{m}{2} + \frac{n-1}{4}$ edges of $G$. Such a 2-coloring can be found in polynomial time.
\end{theorem}
\begin{proof}
Let $V=\{v_1,v_2,\ldots , v_n\}$ and let $c:\ V\rightarrow \{-1,1\}$ be a 2-coloring of $G$. Let $x_p=c(v_p)$, $p\in [n].$
For an edge $v_iv_j\in E$ we set $s_{ij}=1$ if $v_iv_j$ is labeled by $\neq$ and $s_{ij}=-1$ if $v_iv_j$ is labeled by $=$. Then the function $\frac{1}{2}\sum_{v_iv_j\in E}(1-s_{ij}x_ix_j)$ counts the number of edges satisfied by $c$. The rest of the proof is similar to that in the previous theorem.
\qed \end{proof}

\section{Open Problems}\label{sec:op}

Another question of Mahajan et al. \cite{MahajanRamanSikdar09} remains open: what
is the parameterized complexity of deciding whether a connected graph on $n$ vertices and $m$
edges has a bipartite subgraph with at least $m/2+(n-1)/4 + k$ edges, where $k$ is the parameter. Fixed-parameter tractability of a weaker problem was proved by Bollob{\'a}s and Scott \cite{BolSco2002} a decade ago.

The kernel obtained in Theorem \ref{thm:main} is not of polynomial size as it is not polynomial in $m$.
The existence of a polynomial-size kernel for
\textsc{MaxLin2-AA}[$k$] remains an open problem.

Perhaps the kernel obtained in Theorem \ref{thm:main} or the algorithm of Corollary \ref{cor1} can be improved
if we find a structural characterization of irreducible systems for which the maximum excess is less than $2k.$
Such a characterization can be of interest by itself.

Let $F$ be a CNF formula with clauses $C_1,\ldots ,C_m$ of sizes $r_1,\ldots ,r_m$.
Since the probability of $C_i$ being satisfied by a random assignment is $1-2^{-r_i}$, the expected (average)
number of satisfied clauses is $E=\sum_{i=1}^m(1-2^{-r_i}).$ It is natural to consider
the following parameterized problem {\sc MaxSat-AA}[$k$]: decide whether there is a truth assignment that satisfies at least
$E+k$ clauses. When there is a constant $r$ such that $|C_i|\le r$ for each $i=1,\ldots ,m$, {\sc MaxSat-AA}[$k$]
is denoted by {\sc Max-$r$-Sat-AA}[$k$]. Mahajan et al. \cite{MahajanRamanSikdar09} asked what is the complexity of {\sc Max-$r$-Sat-AA}[$k$]
and Alon et al. \cite{AloGutKimSzeYeo11} proved that it is fixed-parameter tractable \cite{AloGutKimSzeYeo11}. It would be interesting to determine the complexity of {\sc MaxSat-AA}[$k$].

\medskip

\paragraph{Acknowledgments}
Research of Crowston, Gutin, Jones and Yeo was partially supported by an International Joint grant of Royal Society.
Fellows is supported by an Australian Research Council Professorial Fellowship.  The research
of Rosamond is supported by an Australian Research Council Discovery Project.
Research of Thomass{\'e} was partially supported by  the AGAPE project (ANR-09-BLAN-0159).

\urlstyle{rm}


\end{document}